\documentclass[conference]{IEEEtran}
\IEEEoverridecommandlockouts

\usepackage{cite}
\usepackage{amsmath,amssymb,amsfonts}
\usepackage{amsthm}
\usepackage{algorithmic}
\usepackage{graphicx}
\usepackage{textcomp}
\usepackage{xcolor}
\usepackage{booktabs}
\usepackage{multirow}

\usepackage{amsthm}

\usepackage{caption}
\newtheorem{theorem}{Theorem}

\newtheorem{proposition}{Proposition}
\newtheorem{corollary}{Corollary}

\theoremstyle{remark}
\newtheorem{remark}{Remark}

\def\BibTeX{{\rm B\kern-.05em{\sc i\kern-.025em b}\kern-.08em
    T\kern-.1667em\lower.7ex\hbox{E}\kern-.125emX}}

\begin{document}

\title{Rethinking Radiomap Blind Prediction with Limited Environment and Configuration Representations}

\author{
	\IEEEauthorblockN{
		Xiaojie Li\IEEEauthorrefmark{1}\IEEEauthorrefmark{2}, 
		Yu Han\IEEEauthorrefmark{1}\IEEEauthorrefmark{2}, 
		Han Fang\IEEEauthorrefmark{3},
        Shangqing Liu\IEEEauthorrefmark{4},
		Shi Jin\IEEEauthorrefmark{1}\IEEEauthorrefmark{2},
		and Chao-Kai Wen\IEEEauthorrefmark{5}
	\IEEEauthorblockA{\IEEEauthorrefmark{1}School of Information Science and Engineering, Southeast University, Nanjing 210096, China}
    \IEEEauthorblockA{\IEEEauthorrefmark{2}National Mobile Communications Research Laboratory, Southeast University, Nanjing 210096, China}
	\IEEEauthorblockA{\IEEEauthorrefmark{3} School of Cyber Science and Engineering, Southeast University, Nanjing 210096, China}
    \IEEEauthorblockA{\IEEEauthorrefmark{4} State Key Laboratory of Novel Software Technology, Nanjing University, Nanjing 210023, China}
	\IEEEauthorblockA{\IEEEauthorrefmark{5}Institute of Communications Engineering, National Sun Yat-sen University, Kaohsiung 80424, Taiwan\\ Email: xiaojieli@seu.edu.cn, hanyu@seu.edu.cn, h\_fang@seu.edu.cn,\\ shangqingliu@nju.edu.cn, jinshi@seu.edu.cn and chaokai.wen@mail.nsysu.edu.tw}}
    \thanks{This work was supported in part by the National Natural Science Foundation of China (NSFC) under Grants 62422105 and U25A20392, in part by the Key Technologies RD Program of Jiangsu (Prospective and Key Technologies for Industry) under Grants BE2023022, BE2023022-1 and BE2023022-2, in part by the Natural Science Foundation of Jiangsu Province under Grant BK20230824, and in part by the Postgraduate Research \& Practice Innovation Program of Jiangsu Province under Grant 26CXJH0640 (Corresponding authors: Yu Han and Shi Jin).}
    }

\maketitle

\begin{abstract}
Radiomap blind prediction infers radiomaps from observable representations of the propagation environment and base station (BS) configuration without field measurements. These representations are inherently incomplete and cannot uniquely determine the target radiomap. Under squared loss, we identify the conditional-mean radiomap as the population-optimal deterministic target and decompose domain risk into target-approximation error and irreducible uncertainty. The train-test risk gap motivates propagation priors as cross-domain guidance, although their partial or simplified forms may bias the attainable predictor. We therefore propose RadioDecomp, which treats a prior-guided predictor as a correctable base and uses deterministic residual refinement to learn its remaining predictable discrepancy. We instantiate RadioDecomp as RadioLSR (LoS-Shadow-Residual). Experiments under cross-configuration and cross-environment settings show that RadioLSR is especially effective for cross-configuration generalization and provides overall gains over a controlled monolithic counterpart under cross-environment generalization.
\end{abstract}

\begin{IEEEkeywords}
radiomap, incomplete observation, generalization, residual learning, wireless propagation
\end{IEEEkeywords}
\section{Introduction}

Future 6G networks are expected to rely on environment-aware radio intelligence, creating strong demand for radiomaps that characterize the spatial distribution of radio attributes such as received power and angle of arrival \cite{zeng2024tutorial}. Such radiomaps are valuable for wireless network planning, optimization, and predictive radio management. In many practical scenarios, the target radiomap must be inferred without field measurements. This gives rise to \emph{radiomap blind prediction}, which predicts radiomaps only from observable representations of the environment and base station (BS) configuration \cite{bi2019engineering,li2026u6gxlmimoradiomapprediction}.

Existing studies have mainly approached this task by learning direct predictive mappings, including environment-driven methods such as RadioUNet \cite{radiounet}, RME-GAN \cite{rmegan}, and RadioDiff \cite{radiodiff}, configuration-aware methods that incorporate BS-side operational parameters \cite{11268973}, and more recent settings with jointly varying environments and BS configurations \cite{li2026u6gxlmimoradiomapprediction,11063460}. Beyond these learning-based approaches, ray-tracing-based simulation \cite{11421037} can be viewed as a physics-driven counterpart when the environment and BS configuration are fully specified. However, only partial environment and configuration representations are available in practical blind prediction.

A related but distinct line of work studies radiomap estimation from sparse radiomap measurements \cite{10542642}. Although noisy environment information is considered in \cite{jaensch2026radiomappredictionnoisy}, such methods still rely on measurement support and therefore do not directly characterize blind radiomap prediction under incomplete observation.

This raises a fundamental question: what does a deterministic blind predictor learn under incomplete observation? Under squared loss, we identify the conditional-mean radiomap as its population-optimal target and decompose domain risk into target-approximation error and irreducible uncertainty, whose domain-wise changes determine the train-test risk gap. Propagation priors provide natural cross-domain guidance for limiting the predictor-dependent approximation gap, but their partial or simplified forms may bias the attainable predictor. We therefore propose \emph{RadioDecomp}, which retains a prior-guided TIC base and uses a deterministic RIC to correct its predictable discrepancy, and instantiate it as \emph{RadioLSR}. Experiments on a multi-configuration and multi-environment radiomap dataset \cite{li2026u6gxlmimoradiomapprediction} show consistent gains under cross-configuration prediction and overall benefits under cross-environment prediction.

\section{Problem Formulation and Analysis}
\subsection{Problem Formulation under Incomplete Observation}

Let \(E\) denote the complete physical environment and \(C\) the complete BS configuration. The former includes propagation-relevant scene information, such as geometry and material properties, while the latter specifies the BS location and transmission setup, including the antenna array, carrier frequency, and beamforming configuration. A radiomap is a discretized spatial representation of a radio attribute, such as received power, angle of arrival, or angle of departure. In the 2D setting considered here, a radiomap realization is represented by \(\mathbf y\in\mathbb R^{H\times W}\); the same formulation can be extended to voxelized 3D radiomaps \cite{10963917,11455177}.

When \(E\) and \(C\) are fully specified, radiomap generation can be expressed as
\begin{equation}
\mathbf y=F(E,C),
\label{eq:full_generation}
\end{equation}
where \(F\) denotes the deterministic physical propagation mechanism.

Radiomap blind prediction differs because the predictor receives only incomplete representations of the environment and BS configuration. Let \(\mathbf Z=(E_{\mathrm{obs}},C_{\mathrm{obs}})\) denote the observable representation, and let \(\mathbf U\) collect the remaining hidden propagation factors. The complete physical state can then be partitioned into \((\mathbf Z,\mathbf U)\), such that
\begin{equation}
\mathbf Y=F(\mathbf Z,\mathbf U).
\label{eq:incomplete_generation}
\end{equation}

For a given observable input \(\mathbf z\), multiple hidden states may be compatible with it and produce different radiomaps. Accordingly,
\begin{equation}
\mathbf Y\mid(\mathbf Z=\mathbf z)
=
F(\mathbf z,\mathbf U),
\qquad
\mathbf U\sim p(\cdot\mid\mathbf z).
\label{eq:conditional_generation}
\end{equation}

Although \(\mathbf U\) is unobserved, a deterministic blind predictor must return one radiomap for each \(\mathbf z\). Under squared loss, its population-optimal output is
\begin{equation}
\begin{aligned}
m^\star(\mathbf z)
&:=
\arg\min_{\widehat{\mathbf y}}
\mathbb E\!\left[
\left\|\mathbf Y-\widehat{\mathbf y}\right\|_{\mathrm F}^{2}
\mid \mathbf Z=\mathbf z
\right]
\\
&=
\mathbb E[\mathbf Y\mid\mathbf Z=\mathbf z]
=
\mathbb E_{\mathbf U\sim p(\cdot\mid\mathbf z)}
\left[F(\mathbf z,\mathbf U)\right].
\end{aligned}
\label{eq:population_conditional_mean}
\end{equation}

Thus, deterministic blind prediction under squared loss seeks to approximate the conditional-mean radiomap \(m^\star(\mathbf z)\), rather than the radiomap associated with any particular realization of the hidden physical factors.

\subsection{Domain-Wise Risk Decomposition and Train-Test Gap}
Let \(d\in\{\mathrm{tr},\mathrm{te}\}\) index the training and test domains, respectively. The population-optimal deterministic predictor in domain \(d\) is
\begin{equation}
m_d^\star(\mathbf z)
:=
\mathbb E_d[\mathbf Y\mid\mathbf Z=\mathbf z]
=
\mathbb E_{\mathbf U\sim p_d(\cdot\mid\mathbf z)}
\left[F(\mathbf z,\mathbf U)\right].
\label{eq:domain_conditional_mean}
\end{equation}

Define the corresponding conditional uncertainty as
\begin{equation}
v_d(\mathbf z)
:=
\mathbb E_d\!\left[
\left\|\mathbf Y-m_d^\star(\mathbf z)\right\|_{\mathrm F}^{2}
\mid\mathbf Z=\mathbf z
\right].
\label{eq:domain_conditional_uncertainty}
\end{equation}

\begin{theorem}[Domain-Wise Risk Decomposition]
\label{thm:domain_risk_decomposition}
For any deterministic predictor \(f\) with finite second moment and any
\(d\in\{\mathrm{tr},\mathrm{te}\}\), the population risk satisfies
\begin{equation}
\begin{aligned}
R_d(f)
&:=
\mathbb E_d\!\left[
\left\|\mathbf Y-f(\mathbf Z)\right\|_{\mathrm F}^{2}
\right]
\\
&=
\underbrace{
\mathbb E_d\!\left[
\left\|f(\mathbf Z)-m_d^\star(\mathbf Z)\right\|_{\mathrm F}^{2}
\right]
}_{\mathcal A_d(f)}
+
\underbrace{
\mathbb E_d[v_d(\mathbf Z)]
}_{\mathcal U_d},
\end{aligned}
\label{eq:domain_risk_decomposition}
\end{equation}
where \(\mathcal A_d(f)\) is the domain-specific target-approximation error and \(\mathcal U_d\) is the irreducible uncertainty induced by incomplete observation.
\end{theorem}

\begin{proof}
Adding and subtracting \(m_d^\star(\mathbf Z)\) in the squared error yields an approximation term, an uncertainty term, and a cross term. The cross term vanishes because
\[
\mathbb E_d[
\mathbf Y-m_d^\star(\mathbf Z)
\mid\mathbf Z
]=0.
\]
The remaining two terms are exactly \(\mathcal A_d(f)\) and \(\mathcal U_d\), which proves \eqref{eq:domain_risk_decomposition}.
\end{proof}

Since the same fixed predictor can be evaluated under both domains, subtracting their risk decompositions gives the following result.
\begin{corollary}[Train-Test Risk Gap]
\label{cor:train_test_risk_gap}
For any fixed deterministic predictor \(f\),
\begin{equation}
\begin{aligned}
R_{\mathrm{te}}(f)-R_{\mathrm{tr}}(f)
={}&
\underbrace{
\mathcal A_{\mathrm{te}}(f)-\mathcal A_{\mathrm{tr}}(f)
}_{\text{target-approximation gap}}
\\
&+
\underbrace{
\mathcal U_{\mathrm{te}}-\mathcal U_{\mathrm{tr}}
}_{\text{conditional-uncertainty gap}}.
\end{aligned}
\label{eq:train_test_risk_gap}
\end{equation}
\end{corollary}

Theorem~\ref{thm:domain_risk_decomposition} separates prediction risk into a predictor-dependent approximation term and predictor-independent uncertainty. Corollary~\ref{cor:train_test_risk_gap} further shows that the train-test risk difference is determined by domain-dependent changes in these two terms. If a predictor achieves a small training-domain approximation error, good generalization requires its approximation error not to increase substantially from the training domain to the test domain. By contrast, the conditional-uncertainty gap is determined by the observation and domain distributions and cannot be reduced through predictor design.

\begin{remark}
The squared loss is used to identify the conditional-mean target and obtain the exact risk decomposition. It does not require \(f\) to be trained using squared loss; the analysis applies to any resulting predictor with finite second moment, including one obtained using an \(\ell_1\) training objective.
\end{remark}

Reducing the predictor-dependent train-test gap requires learning relations that remain useful beyond the training samples. The physical propagation mechanism \(F\) provides a natural source of such guidance because its underlying propagation regularities are shared across environments and BS configurations. Accordingly, incorporating propagation knowledge into the predictor may reduce the finite-sample difficulty of discovering transferable relations and help maintain a small approximation gap across domains.

However, under incomplete observation, an implementable propagation prior can describe only selected or simplified aspects of \(F\). Let \(\Phi\) denote such a prior and let \(\mathcal H_{\Phi}\) denote the effective hypothesis class induced by it. Its domain-specific structural approximation bias can be expressed as
\begin{equation}
\mathcal B_{\Phi,d}
:=
\inf_{b\in\mathcal H_{\Phi}}
\mathbb E_d\!\left[
\left\|b(\mathbf Z)-m_d^\star(\mathbf Z)\right\|_{\mathrm F}^{2}
\right].
\label{eq:prior_induced_bias}
\end{equation}
When \(\mathcal B_{\Phi,d}>0\), even the best predictor permitted by the prior-guided base class cannot attain the conditional target. Thus, propagation priors may improve finite-sample generalization while introducing an approximation bias that limits the attainable predictor. This benefit-bias tradeoff motivates a correctable prior-guided predictor.

\subsection{RadioDecomp as a Correctable Prior-Guided Predictor}

Accordingly, let
\(f_{\mathrm{TIC}}\in\mathcal H_{\Phi}\) denote a propagation-prior-guided base predictor, referred to as the transferable interaction component (TIC). It is intended to capture propagation relations that remain useful across domains, but may deviate from the conditional target because of prior-induced bias, finite-sample learning, or imperfect optimization. To retain this propagation guidance while correcting its predictable discrepancy, we introduce \emph{RadioDecomp}:
\begin{equation}
\hat{\mathbf y}
=
f(\mathbf z)
=
f_{\mathrm{TIC}}(\mathbf z)
+
f_{\mathrm{RIC}}
\!\left(
\mathbf z,
f_{\mathrm{TIC}}(\mathbf z)
\right),
\label{eq:radio_decomp}
\end{equation}
where
\(f_{\mathrm{RIC}}(\mathbf z,f_{\mathrm{TIC}}(\mathbf z))\)
denotes the residual interaction component (RIC). The residual predictor exploits observable relations not captured by TIC and uses \(f_{\mathrm{TIC}}(\mathbf z)\) as a propagation-guided prediction anchor.

To formalize its correctability, consider a fixed structured base
\(f_{\mathrm{TIC}}\) and a residual hypothesis class \(\mathcal H_R\). Define
\begin{equation}
\begin{aligned}
\mathcal H_{\mathrm{RD}}(f_{\mathrm{TIC}})
=
\Big\{
f:\,
f(\mathbf z)
={}&
f_{\mathrm{TIC}}(\mathbf z)
\\
&+
r\!\left(
\mathbf z,
f_{\mathrm{TIC}}(\mathbf z)
\right),
\quad
r\in\mathcal H_R
\Big\}.
\end{aligned}
\label{eq:radiodecomp_hypothesis_class}
\end{equation}

\begin{proposition}[Deterministic Residual Correctability]
\label{prop:residual_correctability}
If the zero function belongs to the residual class,
\(0\in\mathcal H_R\), then, for any
\(d\in\{\mathrm{tr},\mathrm{te}\}\),
\begin{equation}
\inf_{f\in\mathcal H_{\mathrm{RD}}(f_{\mathrm{TIC}})}
R_d(f)
\le
R_d(f_{\mathrm{TIC}}).
\label{eq:radiodecomp_no_worse_than_base}
\end{equation}
Moreover, under squared loss, the pointwise population-optimal deterministic residual in domain \(d\) is
\begin{equation}
f_{\mathrm{RIC},d}^\star
\!\left(
\mathbf z,
f_{\mathrm{TIC}}(\mathbf z)
\right)
=
m_d^\star(\mathbf z)
-
f_{\mathrm{TIC}}(\mathbf z).
\label{eq:optimal_residual}
\end{equation}
\end{proposition}

\begin{proof}
Because \(0\in\mathcal H_R\), choosing the zero residual recovers
\(f_{\mathrm{TIC}}\), which proves
\eqref{eq:radiodecomp_no_worse_than_base}. For a fixed base,
\(f_{\mathrm{TIC}}(\mathbf Z)\) is deterministic given \(\mathbf Z\), and hence
\begin{equation}
\begin{aligned}
&\mathbb E_d[
\mathbf Y-f_{\mathrm{TIC}}(\mathbf Z)
\mid\mathbf Z=\mathbf z
]
\\
&\qquad=
m_d^\star(\mathbf z)
-
f_{\mathrm{TIC}}(\mathbf z).
\end{aligned}
\end{equation}
Since the conditional mean minimizes conditional squared loss,
\eqref{eq:optimal_residual} follows.
\end{proof}

Thus, RadioDecomp retains the prior-guided base as a special case while providing a deterministic correction path toward the domain-specific conditional target. The residual branch can correct the predictable approximation bias of the base.
\section{A RadioDecomp Instantiation: RadioLSR}

Following this principle, we instantiate RadioDecomp as \emph{RadioLSR}, where LSR stands for \emph{LoS-Shadow-Residual}. RadioLSR uses LoS-dominant and blockage-dominant shadowing effects as an observable-supported base estimate, rather than as a complete radiomap decomposition, and refines the remaining discrepancy through a residual part. Accordingly, RadioLSR is written as
\begin{equation}
\hat{\mathbf y}
=
f(\mathbf z)
=
f_{\mathrm{base}}(\mathbf z)
+
f_{\mathrm{res}}\!\big(\mathbf z,f_{\mathrm{base}}(\mathbf z)\big),
\label{eq:radiolsr_decomp}
\end{equation}
where \(f_{\mathrm{base}}\) and \(f_{\mathrm{res}}\) are the RadioLSR instantiations of \(f_{\mathrm{TIC}}\) and \(f_{\mathrm{RIC}}\), respectively. In RadioLSR, both predictors are realized at the prediction level with the corresponding spatial masks incorporated into their outputs.

\subsection{Physics-Aware Input Representation}

We instantiate RadioLSR on the U6G XL-MIMO Radiomap dataset \cite{li2026u6gxlmimoradiomapprediction}. For each sample, the observable input consists of a beam map \(\mathbf{B}\in\mathbb{R}^{H\times W}\), a height map \(\mathbf{H}\in\mathbb{R}^{H\times W}\), two directional edge maps \(\mathbf{E}^{(1)},\mathbf{E}^{(2)}\in\mathbb{R}^{H\times W}\) extracted from the height map along the two spatial axes, and a blockage score map \(\mathbf{S}\in\mathbb{R}^{H\times W}\). Let \(k\in\{1,\dots,HW\}\) denote the index of a spatial grid point in the \(H\times W\) observation domain. Then the \(k\)-th element of the beam map is defined as\cite{li2026u6gxlmimoradiomapprediction}
\begin{equation}
B_k
=
\frac{\lambda^2}{(4\pi)^2}P_t
\left|
\mathbf{w}^{H}\mathbf{H}^{\mathrm{LoS}}_k
\right|^2,
\end{equation}
which represents the analytically computed LoS beamforming power at grid point \(k\) under the given BS configuration and beamforming setting. The height map records building heights, while the two directional edge maps provide boundary cues along the horizontal and vertical grid directions.

Prediction is performed only over the valid non-building region. We therefore define a binary valid-region mask \(\mathbf{M}_{\mathrm{valid}}\in\{0,1\}^{H\times W}\), whose \((i,j)\)-th entry equals \(1\) if the corresponding grid cell belongs to the valid prediction region and \(0\) otherwise.

To facilitate the modeling of blockage-aware attenuation, we further construct a blockage score for each grid point \(k\) by sampling \(M\) points along the BS-grid line segment and comparing the direct-ray height with the corresponding building height:
\begin{equation}
S_k
=
\frac{1}{M}\sum_{m=1}^{M}
\Big[
h_{\mathrm{bld}}\!\big(\mathbf p_{\text{BS}}+t_m(\mathbf p_k-\mathbf p_{\text{BS}})\big)
-
h_{\mathrm{ray}}(t_m)
\Big]_+ ,
\label{eq:blockage_score}
\end{equation}
where \(\mathbf p_{\text{BS}}\) and \(\mathbf p_k\) denote the BS position and the horizontal position of grid point \(k\), respectively, \(\{t_m\}_{m=1}^{M}\subset(0,1)\) are uniformly sampled interpolation factors, \(h_{\mathrm{bld}}(\cdot)\) denotes the building height, \(h_{\mathrm{ray}}(t_m)\) is the ray height at interpolation point \(t_m\), and \([x]_+=\max(x,0)\). Thus, \(S_k\) characterizes the blockage severity along the BS-grid path. Based on \(\mathbf{S}\), we define two binary masks \(\mathbf{M}_{\mathrm{LoS}}, \mathbf{M}_{\mathrm{Shd}}\in\{0,1\}^{H\times W}\) as
\begin{equation}
\mathbf{M}_{\mathrm{LoS}}
=
\mathbb{I}(\mathbf{S}\le\epsilon)\odot \mathbf{M}_{\mathrm{valid}},
\qquad
\mathbf{M}_{\mathrm{Shd}}
=
\mathbb{I}(\mathbf{S}>\epsilon)\odot \mathbf{M}_{\mathrm{valid}},
\end{equation}
where \(\epsilon\) is a blockage-score threshold for separating LoS-dominant and blockage-dominant regions, and \(\mathbb{I}(\cdot)\) denotes the element-wise indicator function. Thus, \(\mathbf{M}_{\mathrm{LoS}}\) and \(\mathbf{M}_{\mathrm{Shd}}\) partition the valid region according to the blockage score.

\subsection{RadioLSR Architecture}
RadioLSR consists of three U-Net branches for LoS prediction, shadow prediction, and residual refinement. The LoS and Shadow branches jointly realize the structured base predictor \(f_{\mathrm{base}}\), while the residual branch realizes \(f_{\mathrm{res}}\). 
This design reflects the available input representation: coarse height map errors mainly affect LoS/shadow boundaries, whereas reflection-related patterns are more sensitive to fine geometry and surface orientation and are therefore deferred to residual refinement.

All three branches adopt the same U-Net template summarized in Table~\ref{tab:unet_backbone}, but differ in their input channels, prediction roles, and base channel width \(C_b\). The base branches use restricted inputs so that the base stage focuses on relatively stable and explicitly supported structure, while harder effects are deferred to refinement.

\subsubsection{Base Branch (LoS and Shadow)}
The base stage contains two U-Net branches. The LoS branch is intended to capture the more direct coverage trend induced by the BS configuration. Since this component is already strongly reflected by the beam map, its prediction is defined as
\begin{equation}
\hat{\mathbf y}_{\mathrm{LoS}}
=
f_{\mathrm{LoS}}(\mathbf B)\odot \mathbf M_{\mathrm{LoS}}.
\end{equation}
This restriction helps keep the branch focused on the most stable part of the structured base predictor.

The Shadow branch is intended to capture blockage-aware attenuation, which remains structured but depends more explicitly on obstruction cues. Its prediction is defined as
\begin{equation}
\hat{\mathbf y}_{\mathrm{Shd}}
=
f_{\mathrm{Shd}}([\mathbf B,\mathbf S])\odot \mathbf M_{\mathrm{Shd}}.
\end{equation}
Here, the blockage score provides a compact cue for obstruction severity, while the beam map preserves the configuration-dependent coverage tendency. This restricted input design encourages the branch to focus on blockage-related attenuation.

The structured base predictor is then defined as
\begin{equation}
f_{\mathrm{base}}(\mathbf z)
:=
\hat{\mathbf y}_{\mathrm{LoS}}+\hat{\mathbf y}_{\mathrm{Shd}}.
\label{eq:radiolsr_base}
\end{equation}
Equivalently,
\begin{equation}
f_{\mathrm{base}}(\mathbf z)
=
f_{\mathrm{LoS}}(\mathbf B)\odot \mathbf M_{\mathrm{LoS}}
+
f_{\mathrm{Shd}}([\mathbf B,\mathbf S])\odot \mathbf M_{\mathrm{Shd}}.
\end{equation}
Thus, \(f_{\mathrm{base}}\) realizes the TIC stage in RadioLSR through masked aggregation of the two base-branch predictions.

\subsubsection{Residual Refinement}
The residual stage contains one additional U-Net branch that refines the remaining discrepancy in the dB domain. Let the residual U-Net branch be denoted by \(g_{\mathrm{res}}(\cdot)\). We define the residual predictor as
\begin{equation}
f_{\mathrm{res}}(\mathbf z,f_{\mathrm{base}}(\mathbf z))
:=
g_{\mathrm{res}}([\mathbf B,\mathbf H,\mathbf E^{(1)},\mathbf E^{(2)},f_{\mathrm{base}}(\mathbf z)])
\odot \mathbf M_{\mathrm{valid}}.
\label{eq:radiolsr_res}
\end{equation}
Accordingly, the final prediction is
\begin{equation}
\hat{\mathbf y}
=
f_{\mathrm{base}}(\mathbf z)
+
f_{\mathrm{res}}(\mathbf z,f_{\mathrm{base}}(\mathbf z)).
\end{equation}

Conditioning the residual U-Net on \(f_{\mathrm{base}}(\mathbf z)\) is consistent with the additive composition in the dB domain, so refinement is learned relative to the current attenuation level rather than as an independent full-map prediction. The residual branch is not additionally fed with blockage score, since blockage-aware effects have already been assigned to the base stage.

\begin{table}[t]
\vspace{4pt}
\centering
\caption{Unified U-Net backbone of the RadioLSR branches.}
\label{tab:unet_backbone}
\setlength{\tabcolsep}{4pt}
\begin{tabular}{c|c}
\toprule
\textbf{Stage} & \textbf{Specification} \\
\midrule
Encoder & DoubleConv $(C_b)$, then 3$\times$[MaxPool + DoubleConv] \\
        & with channels $2C_b,\,4C_b,\,8C_b$ \\
Bottleneck & DoubleConv $(8C_b)$ \\
Decoder & 3$\times$[Up + skip concat + DoubleConv] \\
        & with channels $4C_b,\,2C_b,\,C_b$, then $1{\times}1$ Conv $(1)$ \\
\bottomrule
\end{tabular}
\end{table}

\subsection{Training Objective}

Let \(\mathbf y\) denote the ground-truth normalized radiomap. We train RadioLSR with a masked \(\ell_1\) loss on the final prediction together with auxiliary supervision on the two base branches (LoS and Shadow) and the residual branch:
\begin{equation}
\mathcal L
=
\mathcal L_{\mathrm{main}}
+\lambda_{\mathrm{LoS}}\mathcal L_{\mathrm{LoS}}
+\lambda_{\mathrm{Shd}}\mathcal L_{\mathrm{Shd}}
+\lambda_{\mathrm{res}}\mathcal L_{\mathrm{res}}.
\end{equation}

where\footnote{For numerical stability, the denominator of each masked-average loss is clamped to be at least \(1\).}
\(\hat{\mathbf y}_{\mathrm{res}}=f_{\mathrm{res}}(\mathbf z,f_{\mathrm{base}}(\mathbf z))\), and
{\small
\begin{equation}
\begin{aligned}
\mathcal L_{\mathrm{main}}
&={\|(\hat{\mathbf y}-\mathbf y)\odot\mathbf M_{\mathrm{valid}}\|_1}/{\|\mathbf M_{\mathrm{valid}}\|_1},\\
\mathcal L_{\mathrm{LoS}}
&={\|(\hat{\mathbf y}_{\mathrm{LoS}}-\mathbf y)\odot\mathbf M_{\mathrm{LoS}}\|_1}/{\|\mathbf M_{\mathrm{LoS}}\|_1},\\
\mathcal L_{\mathrm{Shd}}
&={\|(\hat{\mathbf y}_{\mathrm{Shd}}-\mathbf y)\odot\mathbf M_{\mathrm{Shd}}\|_1}/{\|\mathbf M_{\mathrm{Shd}}\|_1},\\
\mathcal L_{\mathrm{res}}
&={\|(\hat{\mathbf y}_{\mathrm{res}}-\mathbf y_{\mathrm{res}}^\star)\odot\mathbf M_{\mathrm{valid}}\|_1}/{\|\mathbf M_{\mathrm{valid}}\|_1}.
\end{aligned}
\end{equation}}

The residual target is defined as
\begin{equation}
\mathbf y_{\mathrm{res}}^\star
=
(\mathbf y-\mathrm{sg}[f_{\mathrm{base}}(\mathbf z)])\odot \mathbf M_{\mathrm{valid}},
\end{equation}
where \(\mathrm{sg}[\cdot]\) denotes the stop-gradient operation. Here, \(\mathcal L_{\mathrm{main}}\) supervises the final composed prediction, while the other three terms encourage LoS, shadow, and residual specialization. Optimization is performed using AdamW with a validation-based ReduceLROnPlateau scheduler \cite{al2022scheduling}.

\section{Experiments}
The experiments evaluate RadioLSR and examine whether a RadioDecomp-guided structured predictor generalizes better than a monolithic counterpart under incomplete observation.
\subsection{Experimental Setup}
Experiments are conducted on the U6G XL-MIMO Radiomap dataset \cite{li2026u6gxlmimoradiomapprediction}, which contains 78,400 radiomaps from 800 urban scenes and 98 BS configurations. We consider two blind-prediction protocols: \emph{cross-config}, which splits data by BS configuration, and \emph{cross-env}, which splits data by environment. For both protocols, six train/validation/test ratios are used: $2/1/7$, $3/1/6$, $4/1/5$, $5/1/4$, $6/1/3$, and $7/1/2$.
\begin{table*}[!t]
\vspace{4pt}
\centering
\caption{Train MAE and test MAE/RMSE (dB) of RadioLSR-32 and MonoUNet-32 under the cross-config and cross-env splits.}
\label{tab:main_results_combined}
\setlength{\tabcolsep}{3.2pt}
\begin{tabular}{c|ccc|ccc|ccc|ccc}
\toprule
\multirow{3}{*}{\textbf{Train/Val/Test}}
& \multicolumn{6}{c|}{\textbf{Cross-config split}}
& \multicolumn{6}{c}{\textbf{Cross-env split}} \\
\cline{2-13}
& \multicolumn{3}{c|}{\textbf{RadioLSR-32}}
& \multicolumn{3}{c|}{\textbf{MonoUNet-32}}
& \multicolumn{3}{c|}{\textbf{RadioLSR-32}}
& \multicolumn{3}{c}{\textbf{MonoUNet-32}} \\
\cline{2-13}
& \textbf{Train}
& \textbf{Test}
& \textbf{Test}
& \textbf{Train}
& \textbf{Test}
& \textbf{Test}
& \textbf{Train}
& \textbf{Test}
& \textbf{Test}
& \textbf{Train}
& \textbf{Test}
& \textbf{Test} \\
& \textbf{MAE}
& \textbf{MAE}
& \textbf{RMSE}
& \textbf{MAE}
& \textbf{MAE}
& \textbf{RMSE}
& \textbf{MAE}
& \textbf{MAE}
& \textbf{RMSE}
& \textbf{MAE}
& \textbf{MAE}
& \textbf{RMSE} \\
\midrule
2/1/7 & \textbf{1.8044} & \textbf{3.5673} & \textbf{5.8391} & 2.2042 & 3.8999 & 6.4388 & \textbf{2.3412} & \textbf{4.3533} & \textbf{8.4919} & 2.4091 & 4.4039 & 8.5521 \\
3/1/6 & \textbf{1.8543} & \textbf{3.5525} & \textbf{5.6258} & 2.0091 & 4.2193 & 6.6179 & \textbf{2.3224} & \textbf{4.1862} & 8.5254 & 2.3499 & 4.3337 & \textbf{8.4838} \\
4/1/5 & \textbf{1.6079} & \textbf{3.6721} & \textbf{5.5804} & 1.9079 & 4.0620 & 6.2850 & 2.6196 & \textbf{4.1241} & \textbf{8.1609} & \textbf{2.4224} & 4.1252 & 8.2251 \\
5/1/4 & \textbf{1.5868} & \textbf{3.7785} & \textbf{5.6960} & 1.8445 & 4.3276 & 6.3687 & 2.4463 & \textbf{4.0494} & \textbf{8.0615} & \textbf{2.0895} & 4.2097 & 8.4318 \\
6/1/3 & \textbf{1.5893} & \textbf{4.1106} & \textbf{5.8316} & 1.7914 & 4.3287 & 6.3111 & \textbf{2.3853} & \textbf{3.9533} & \textbf{7.9311} & 2.4924 & 4.0207 & 8.0607 \\
7/1/2 & \textbf{1.6089} & \textbf{4.8613} & \textbf{6.4511} & 1.9006 & 5.0672 & 6.9664 & 2.1774 & \textbf{4.0771} & \textbf{8.1701} & \textbf{1.9845} & 4.2538 & 8.6370 \\
\bottomrule
\end{tabular}
\end{table*}
\begin{figure}[t]
    \centering
\includegraphics[width=0.6\linewidth]{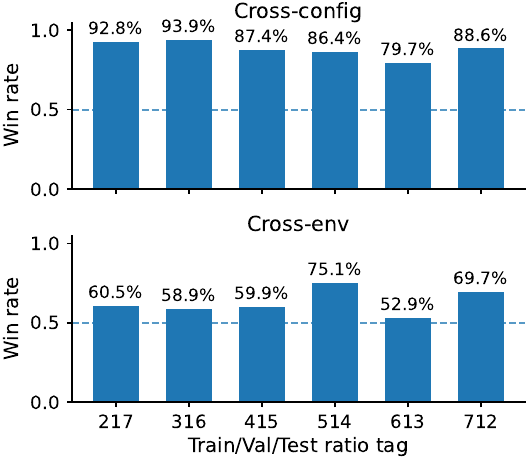}
\caption{Sample-level win rate of RadioLSR-32 over MonoUNet-32 under six train-ratio settings for the cross-config and cross-env splits.}
    \label{fig:placeholder}
\end{figure}

\begin{table*}[!t]
\centering
\caption{Aggregated test-sample RMSE statistics (dB) over all considered train-ratio settings under the cross-config and cross-env splits, where P50, P90, and P95 denote the 50th, 90th, and 95th percentiles, respectively.}
\label{tab:rmse_dist_beam_scene}
\setlength{\tabcolsep}{4pt}
\begin{tabular}{l|ccccc|ccccc}
\toprule
\multirow{2}{*}{\textbf{Method (Parameter Num)}}
& \multicolumn{5}{c|}{\textbf{Cross-config split} (test \(217{,}600\) radiomaps)}
& \multicolumn{5}{c}{\textbf{Cross-env split} (test \(211{,}680\) radiomaps)} \\
\cline{2-11}
& \textbf{Mean} & \textbf{Std} & \textbf{P50} & \textbf{P90} & \textbf{P95}
& \textbf{Mean} & \textbf{Std} & \textbf{P50} & \textbf{P90} & \textbf{P95} \\
\midrule
RadioLSR-32 (4.6960M) & \textbf{5.5438} & \textbf{1.8636} & \textbf{5.5328} & \textbf{7.8895} & \textbf{8.5605} & 7.4484 & 4.1423 & 6.7295 & 12.9398 & 14.6668 \\
RadioLSR-48 (8.6039M) & 5.6205 & 1.8905 & 5.5497 & 7.9348 & 8.7384 & \textbf{7.4001} & \textbf{4.0835} & \textbf{6.7244} & \textbf{12.8139} & \textbf{14.5418} \\
MonoUNet-32 (3.1296M) & 6.2092 & 2.1128 & 6.1795 & 8.9076 & 9.7229 & 7.6014 & 4.1401 & 6.8462 & 13.4127 & 15.0610 \\
MonoUNet-48 (7.0375M) & 5.7783 & 1.9786 & 5.7812 & 8.2649 & 9.0100 & 7.5080 & 4.0797 & 6.7620 & 13.2766 & 14.9046 \\
MonoUNet-64 (12.5076M) & 5.6770 & 1.9618 & 5.6744 & 8.1225 & 8.8935 & 7.4873 & 4.1230 & 6.7381 & 13.2967 & 14.9911 \\
\bottomrule
\end{tabular}
\end{table*}

\begin{table}[t]
\centering
\caption{Ablation under 6/1/3 split, comparing full RadioLSR with its TIC-only variant in the cross-config and cross-env settings.}
\label{tab:ablation_tic_only_613}
\setlength{\tabcolsep}{4pt}
\begin{tabular}{c|cc|cc}
\toprule
\multirow{2}{*}{\textbf{Method}}
& \multicolumn{2}{c|}{\textbf{Cross-config split}}
& \multicolumn{2}{c}{\textbf{Cross-env split}} \\
\cline{2-5}
& \textbf{MAE} & \textbf{RMSE}
& \textbf{MAE} & \textbf{RMSE} \\
\midrule
RadioLSR-32 & \textbf{4.1106} & \textbf{5.8316} & \textbf{3.9533} & \textbf{7.9311} \\
TIC only    & 5.5327 & 8.8174 & 4.9422 & 10.4024 \\
\midrule
Error reduction & 1.4221 & 2.9858 & 0.9889 & 2.4713 \\
\bottomrule
\end{tabular}
\end{table}

\begin{figure*}[!t]
    \centering
    \includegraphics[width=0.85\linewidth]{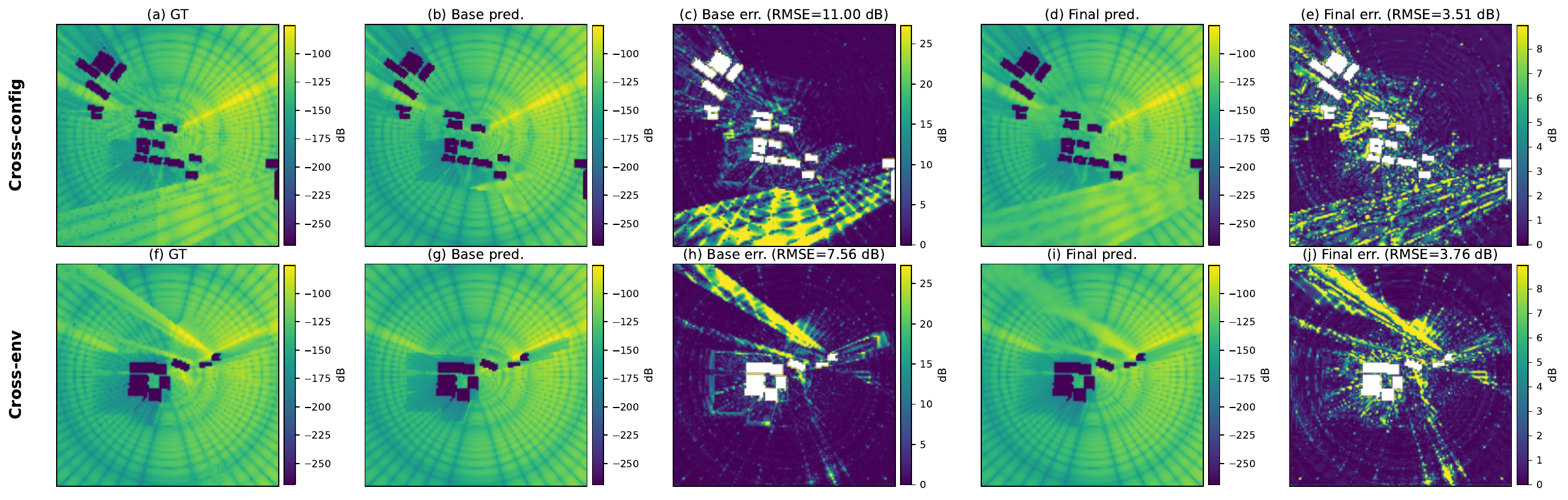}
\caption{Qualitative results of one radiomap sample for RadioLSR under the cross-configuration and cross-environment splits, including the ground truth, base prediction, base error, final prediction, and final error.}
\label{fig:radiolsr_qualitative}
\end{figure*}
Alongside RadioLSR, we also construct a monolithic counterpart, termed \emph{MonoUNet}, to provide a controlled reference under the same input information and closely related backbone family. Both models use the same observable inputs, including the beam map, height map, edge features, and blockage score. For RadioLSR, the LoS and Shadow branches use a base channel width $C_b$ in Table~\ref{tab:unet_backbone} of 16, while the residual-branch width follows the model suffix; thus, RadioLSR-32 and RadioLSR-48 use residual widths of 32 and 48, respectively. For MonoUNet, the suffix denotes the base channel width of the monolithic U-Net, yielding MonoUNet-32, MonoUNet-48, and MonoUNet-64. MonoUNet is trained with the same masked \(\ell_1\) radiomap loss together with an additional Sobel-based structural loss to preserve spatial structure.

The blockage threshold is set to \(\epsilon=10^{-6}\). For RadioLSR, \(\lambda_{\mathrm{LoS}}=\lambda_{\mathrm{Shd}}=\lambda_{\mathrm{res}}=0.5\). All models are trained for 60 epochs with initial learning rate $10^{-3}$, and the best validation checkpoint is used for evaluation. This does not conflict with Section II: the squared loss there defines the conditional mean predictor, while the training losses here only determine the learned predictor \(f(\mathbf z)\). Accordingly, Mean Absolute Error (MAE) is emphasized for optimization consistency, while Root Mean Square Error (RMSE) is also reported to reflect the squared-loss perspective in the theoretical analysis.

\subsubsection{RadioDecomp-Guided Main Comparison}
We first compare the RadioDecomp-guided model RadioLSR-32 with its monolithic counterpart MonoUNet-32 under matched inputs and closely related backbones. Table~\ref{tab:main_results_combined} shows that RadioLSR-32 maintains comparable or lower training MAE in most settings, indicating that structured reparameterization does not weaken train-domain fitting. On the test sets, it consistently achieves lower MAE and RMSE under cross-config and overall better performance under cross-env, with only limited exceptions. Fig.~\ref{fig:placeholder} further supports this trend: RadioLSR wins on a clear majority of samples in all cross-config settings and remains above 50\% in all cross-env settings.

\subsubsection{Capacity and Reparameterization Analysis}
We next examine whether the gain can be explained solely by increasing monolithic capacity. Table~\ref{tab:rmse_dist_beam_scene} shows that widening MonoUNet from 32 to 48 and 64 generally improves RMSE, especially under cross-config. However, the best overall results are still achieved by RadioLSR-32 or RadioLSR-48 rather than by wider MonoUNet baselines. The structured models maintain lower mean RMSE under both splits, with advantages also in P90 and P95. Although the gaps are modest, the comparison uses identical observable inputs and a large number of test radiomaps, supporting the benefit of structured reparameterization beyond monolithic scaling.

\subsubsection{Ablation on the Role of the RIC Refinement}
To examine the refinement role implied by RadioDecomp, we compare the full RadioLSR-32 model with its TIC-only variant under the 6/1/3 split. As shown in Table~\ref{tab:ablation_tic_only_613}, removing RIC causes clear degradation under both split protocols: the full model reduces MAE/RMSE by 1.4221/2.9858 dB under cross-config and by 0.9889/2.4713 dB under cross-env. This indicates that, within the RadioDecomp view, structured base estimation alone is insufficient, while the RIC stage provides an essential refinement path for the remaining discrepancy.

\subsubsection{Qualitative Analysis of the Structured Prediction Path}
Fig.~\ref{fig:radiolsr_qualitative} provides qualitative evidence for the structured prediction path motivated by RadioDecomp. In both cross-config and cross-env examples, the LoS-Shadow base captures the dominant large-scale coverage pattern, while the residual branch contributes localized corrections in more difficult regions, including reflection-related variations supported by the available geometric cues. This behavior is consistent with the intended roles of base estimation and residual refinement in RadioLSR.

\section{Conclusion}
This paper revisited deterministic radiomap blind prediction under incomplete observation. Under squared loss, we identified the conditional-mean radiomap as its population-optimal target, decomposed domain risk into approximation error and irreducible uncertainty, and characterized the train-test risk gap through their domain-wise changes. Since propagation priors offer cross-domain guidance but may bias the attainable predictor, we proposed RadioDecomp, which combines a prior-guided TIC base with a deterministic RIC, and instantiated it as RadioLSR. Experiments showed that RadioLSR is especially effective for cross-config generalization and provides overall benefits under cross-env generalization over a controlled monolithic reference. Future work may include generative RIC modeling: its conditional mean could retain base-to-target correction while its samples represent the predictive variability induced by incomplete observation.

\bibliographystyle{IEEEtran}
\bibliography{bibtex/bib/IEEEexample}

\end{document}